\documentclass[11pt, letter]{article}

\usepackage[english]{babel}
\usepackage[T1]{fontenc}
\usepackage{stmaryrd}
\usepackage{amssymb}
\usepackage{marginnote}
\usepackage{booktabs,tabularx,threeparttable}
\usepackage[normalem]{ulem}

\usepackage{tikz}
\usetikzlibrary{fit,backgrounds,positioning}
\usepackage{fullpage}
\usepackage{mathrsfs}
\usepackage{amsmath}
\usepackage{bbm}
\usepackage{graphicx}
\usepackage{stackengine}
\usepackage{amsthm}
\usepackage{subcaption}
\usepackage{authblk}

\usepackage{thmtools}
\usepackage{thm-restate}

\usepackage[colorlinks=true, allcolors=blue]{hyperref}

\usepackage{fancyhdr,lastpage}
\usepackage{amsfonts}
\usepackage[dvipsnames]{xcolor}

\usepackage{enumitem}

\usepackage{color}
\usepackage{hyperref}

\usepackage{algorithm}
\usepackage{algpseudocode}
\floatname{algorithm}{Auction}

\usepackage{amssymb,amsmath}
\usepackage{tikz}

\renewcommand\qedsymbol{$\blacksquare$}

\def\val[#1,#2](#3,#4){
    \fill[draw=black,color=#2] (#3-0.5,#4-0.5) rectangle (#3+0.5,#4+0.5);
    \draw (#3-0.5,#4-0.5) rectangle (#3+0.5,#4+0.5);
    \node at (#3,#4) {#1}
}
            
\def\valDetailed[#1,#2](#3,#4)[#5,#6,#7,#8]{
    \fill[draw=black,color=#2] ({#3*3},{#4*2}) rectangle ({#3*3+2},{#4*2+2});
    \draw ({#3*3},{#4*2}) rectangle ({#3*3+2},{#4*2+2});
    \node at ({#3*3+1},{#4*2+1}) {#1};
        
    \ifthenelse{\equal{#5}{out} \OR \equal{#5}{clinched_out}}{
        \draw ({#3*3},{#4*2}) -- ({#3*3+2},{#4*2+2});
    }{}
    \ifthenelse{\equal{#5}{clinched} \OR \equal{#5}{clinched_out}}{
        \node at ({#3*3+1},{#4*2+1.6}) {\tiny $Q_{#6}$ in $M_{#7}$};
        \node at ({#3*3+1},{#4*2+0.3}) {\tiny at $p=#8$};
    }{}
}

\newtheorem{theorem}{Theorem}
\newtheorem{claim}[theorem]{Claim}

\newtheorem{lemma}[theorem]{Lemma}
\newtheorem{definition}[theorem]{Definition}

\newtheorem{corollary}[theorem]{Corollary}
\newtheorem{example}{\tt Example}

\usetikzlibrary{arrows,shapes,decorations.pathmorphing,automata,backgrounds,petri}

\newcommand{\es}{\emptyset}
\newcommand{\sm}{\setminus}
\newcommand\paren[1]{\left(#1\right)}

\begin{document}
\def\UrlBreaks{\do\/\do-\do\.\do\_\do\?\do\&\do\=\do\#}

\title{The Tiered Clinching Auction\\ with Applications to Carbon Offset Markets}

\author[1]{Tessa Davis}
\author[1]{Agn\`es Totschnig\thanks{Supported by FRQNT Grant 332481.}}
\author[1]{Adrian Vetta\thanks{Supported by NSERC Discovery Grant 2022-04191.}}
\affil[1]{McGill University, Montreal, Canada}
\affil{\texttt{\{tessa.davis, agnes.totschnig\}@mail.mcgill.ca, adrian.vetta@mcgill.ca}}

\date{}

\maketitle

\begin{abstract}
Voluntary carbon offsetting is a strategy which has been pursued globally by corporations to reduce their effective carbon emissions. Carbon offset markets currently suffer from low liquidity: offsets are highly heterogeneous, the landscape of accreditation is fragmented, and the market lacks a centralized trading infrastructure which would provide clear demand and price signaling to producers and buyers of offsets. In this paper, we propose an ascending auction mechanism which can be applied to the sale of carbon offsets. It generalizes Ausubel's clinching auction to a setting with items of tiered quality. The central idea behind the clinching mechanism is that bidders are allocated items when their opponents' demand drops below the supply. This is generalized to multiple nested submarkets where clinching can occur in the tiered case. Assuming that bidders have a minimum quality level that they will accept but are indifferent to quality beyond that, along with having decreasing marginal utility, we obtain that the auction generates the efficient outcome and charges VCG prices. As a result, sincere bidding is a weakly dominant strategy given private values. Moreover, the auction can be implemented in polynomial time. Beyond offset markets, this auction can be applied to any market where items can be ordered on a scale and bidders have a cutoff point for desiring items on the scale, such as hotel room bookings, concert ticketing and sponsored search auctions.
\end{abstract}

\section{Introduction}
\label{sec:intro}
There has been significant interest from economists, computer scientists, and operations researchers in developing auctions for the sale of various commodities. Auctions can have an advantage over other mechanisms due to their ability to enable price discovery and increase market efficiency. One application of auctions is carbon markets, which price allowances to emit carbon. A related mechanism for emission reductions is carbon offsets. Carbon offset markets are fragmented and relatively inefficient, and are also more complex than carbon permit markets due to the heterogeneity of the offsets. With the goal of improving market function by allowing sellers to get clearer demand and price signaling, this paper proposes an ascending auction mechanism to price carbon offsets by extending the clinching auction of Ausubel \cite{Ausubel2004}.

Our mechanism, the {\em tiered clinching auction}, is efficient, incentive compatible, and can be implemented in polynomial time.
This auction can be applied to any setting where goods can be ordered in tiers according to some underlying quality scale.
For carbon offsets, the tiers are scaled with respect to the credibility of the offsets.  
Other examples include: hotel rooms where tiers are based upon hotel star rankings;
transportation bookings where tiers are based upon service class;
sporting events and rock concerts where tiers are based upon the viewing quality of the seat; and position auctions, such as sponsored search auctions, where tiers are based upon the location of the advertisement.

\subsection{Background: Carbon Offsets and Permits}
The motivation for this work concerns carbon offsets. In 
brief, existing marketplaces for carbon offsets are disparate and suffer from a severe lack of credibility and transparency. This necessitates a more effective and trustworthy market mechanism for carbon offset trading.

To understand this issue, observe that since the 1990s there has been interest in pricing carbon, in a bid to reduce emissions and curb climate change. As countries vowed to reduce their emissions in the Kyoto Protocol, a new mechanism emerged to help them meet their targets: carbon offsets, which are generated by projects that reduce, avoid, or capture a certain amount of carbon emissions. These projects, initially established by the Clean Development Mechanism \cite{GS2012}, include, for example, transitioning to renewable energies and carbon capture, with particular interest in limiting deforestation.

The other primary methods used to price carbon are carbon taxes and cap-and-trade systems. In cap-and-trade systems, countries set a limit on the amount of carbon that can be emitted and issue (or auction) permits called carbon credits to companies, which allow them to emit a certain amount of carbon. Companies which emit less than their allowance can sell their excess credits to other companies which are polluting above their allowance. These kinds of governmentally regulated carbon markets are called emission trading systems (ETS) and commonly referred to as {\em compliance markets}. They have been implemented in a number of places: notable examples include the European Union's international ETS \cite{ZW2010}, the Western Climate Initiative (including California, Washington and Qu\'ebec) and China's ETS, the latter of which covers one-seventh of global carbon emissions from fossil-fuel combustion \cite{IEA2020}.

But beyond compliance markets, there is interest from companies branding themselves as carbon neutral or ``net-zero'' by reducing and offsetting their emissions, and from individuals compensating their travels and other polluting activities. This has led to the development of voluntary carbon offset markets, which have grown rapidly over the past decade to reach \$1 billion dollars in global transactions in 2021~\cite{CMMS2021} and which are expected to increase by up to a factor of 100 by 2050 \cite{BLMP2021}. As diverse projects have emerged, there have also been a number of initiatives developing standards certifying the effectiveness and social impact of offset projects. Notable examples include the Gold Standard, Verified Emission Reduction (VER) and the Verified Carbon Standard (VCS),
which have been surveyed by Kollmuss et al \cite{KZP2008}.

However, the lack of a standardized regulatory framework means that carbon credits are ``highly heterogeneous'' \cite{BLMP2021} and the markets for trading them are ``considerably fragmented''~\cite{KZP2008}. 
As the quality and impact of carbon offsets have raised important concerns, diverse stakeholders have recently proposed new measures to ensure the integrity of voluntary carbon markets:  
the announcement of a collaboration between six independent crediting programs aiming to ``enhance transparency and consistency across the market''~\cite{Verra2023},
six joint recommendations put forward at COP28 by a group of European countries \cite{COP28}, 
10 Core Carbon Principles (CCPs) established by the multi-stakeholder-led independent Integrity Council for the Voluntary Carbon Market (ICVCM) 
and 21 good practices formulated in a consultation report by the International Organization of Securities Commissions (IOSCO) \cite{IOSCO2023}.

McKinsey \cite{BLMP2021} identified the varying accreditation standards, as well as a lack of a centralized market, as two of the central issues behind the low liquidity of voluntary carbon markets. Efficiently scaling these markets to meet the rising demand would require establishing a ``resilient, flexible [trading] infrastructure'' \cite{BLMP2021}. This would moreover address the lack of clear demand signaling and uncertain prices which impact offset producers' ability to commit to costly projects.

In this paper we design an auction mechanism to address the specificities of these markets. We remark it is natural
to use auctions in carbon pricing markets; see for example
the Western Climate Initiative. 
Indeed Weishaar \cite{Weishaar2009} provides an extensive overview of the economic and legal aspects of auction-based mechanisms used in emission trading systems. Cramton and Kerr \cite{CK2002} make a strong case for auctioning newly allocated carbon credits, rather than grandfathering permits based on a company's previous emissions. They propose an ascending clock auction due to its advantages in price discovery. Cong and Wei's \cite{CW2012} experimental analysis suggests that ascending auctions can offer better efficiency than uniform-price auctions, although they may also be less robust to collusion depending on the market conditions. Goldner et al. \cite{GIL2020} propose a uniform price auction with a price floor and ceiling which obtains a constant factor approximation to the optimal social welfare at any equilibrium, under some market constraints.
However, all of these auctions pertain to the carbon credits issued in compliance markets, which are homogeneous. Our paper aims to address the gap in the literature on voluntary carbon offset markets, where offsets are notably heterogeneous.

We emphasize that our aim is to show that an efficient and transparent marketplace can be designed, in theory, for carbon offsets.
We do not address the issue of the credibility of the offsets sold in such a market or who would evaluate the quality of offsets; for that, we refer the reader to the policy recommendations given by the governmental organisations listed above.

\subsection{Background: Auction Theory}
A central concern of auction design is to ensure an efficient outcome, achieved by assigning all items in a way that maximizes social welfare. However, as shown by Rothkopf et al. \cite{RPH1998}, computing the optimal allocation in the multi-unit setting (which is called the winner determination problem) is NP-hard when no additional assumptions are made about the utility functions of bidders, who can place bids on any combination of items. We refer the reader to the book by Cramton, Shoham and Steinberg \cite{CSS2006} for an extensive overview of combinatorial auctions.

Despite this hardness result, there has been significant interest in developing auctions in a broad range of more specific settings, motivated by applications such as airport time slots scheduling, spectrum auctions, and industrial procurement. In practice, ascending auctions, most commonly known as English auctions, where bidders reveal information about their preferences by giving open bids, are popular due to their ability to enable price discovery and stimulate competition. This can increase efficiency and lower barriers to entry \cite{Cramton1998}.

One of the special cases most relevant to this paper is the clinching auction proposed by Ausubel \cite{Ausubel2004}. Its setting assumes that items are homogeneous and restricts bidders' preferences to diminishing marginal utilities. In this model, Ausubel is able to obtain the efficient result of the Vickrey-Clarke-Groves (VCG) mechanism with an ascending auction format.
In the more general setting of heterogeneous items, but with bidders restricted to unit-demand, Demange et al. \cite{DGS1986} achieve the efficient outcome by considering sets of over-demanded items.

A more general model of bidders' utilities is the assumption of gross substitutes, where bidders can demand multiple (heterogeneous) items, but the increase in price of a given item must weakly increase demand for any other item.
Gul and Stacchetti \cite{GS2000} showed that no dynamic auction can reveal enough information to implement the VCG outcome under gross substitutes preferences, concluding that the above results are the most general settings which have an efficient and strategy-proof dynamic auction mechanism.

To overcome these limitations, looking beyond ascending auctions, Ausubel~\cite{Ausubel2006} proposes a mechanism with parallel auctions, which credit and debit items to bidders in a similar fashion to his clinching auction for homogeneous items. This achieves a modified VCG outcome, differing by only an additive constant, which preserves strategy-proofness, but only guarantees a result in finite time. 
Under the weaker assumption of submodularity, but giving up on anonymity and additivity of prices,
de Vries et al. \cite{VSV2007} achieve the VCG outcome in exponential time. Specifically, their mechanism can propose an individualized price $p_i(S)$ to agent $i$ for a subset of items $S$.
This generalizes the work of Demange et al. \cite{DGS1986} by considering minimally undersupplied sets of bidders and applying a linear programming approach introduced by Bikhchandani et al. \cite{BO2002,BVSV2002}.

In this paper, we consider a special case of gross substitute preferences and heterogeneous items. By ordering the items on a scale of quality and restricting the form of bidders' valuations, we are able to obtain very strong results. We develop a strictly ascending auction with VCG outcomes and anonymous prices. Importantly, the auction can be implemented in polynomial time.

Note that our feasibility constraints define a structured, nested integer polymatroid. By a reduction to matroids, Bikhchandani et al. \cite{BDSV2011} give a general ascending VCG auction, which only achieves pseudopolynomial time. 
For polymatroidal environments with more restricted valuation functions, including applications to AdWords auctions, Goel et al. \cite{GMP2015} give a polynomial time clinching mechanism.

\subsection{The Model}\label{sec:model}
There are $L$ tiers (or levels) of items which are ordered on a scale of quality. The tiers are denoted $1, \dots, L$, with $1$ being the lowest quality and $L$ the highest. Within each tier, there is a set of homogeneous items. The quantities of items in each tier are denoted $q_1, \dots, q_L$, and $q=\sum_{k=1}^L q_k$ denotes the total number of items. These quality rankings represent assessments of the quality of carbon offsets, where the evaluation may incorporate the degree of certainty in the additionality of the offset, and any social co-benefits the offset has.

There are $n$ bidders. Each bidder has a lowest quality tier they will accept, and they will, of course, accept items of any quality above that tier. Their valuation for items in lower tiers is zero. Otherwise, bidders' valuations do not depend on the quality of the items they receive. For bidder $i$, if the lowest tier they will accept is $t$, we say that $i$'s tier or level is $t$, denoted $\ell(i) = t$. Given that the bidders represent corporations interested in achieving net zero, it is plausible that their valuation functions can be modeled this way where they have a minimum bar for acceptable offset quality, but do not care about the quality beyond that it is `good enough'.

Each bidder has concave valuations, given by a non-increasing, non-negative marginal valuation function. We denote the marginal valuation of bidder $i$ for their $j$\textsuperscript{th} item as $u_i(j)$. A bidder's utility is equal to the sum of their marginal valuations for each item they received, minus the price paid.

For the sake of simplicity, we assume that all marginal valuations among all bidders are distinct to avoid the problem of having to choose a tie-breaking rule. Otherwise, one can impose any ordering by slightly perturbing bidders' valuations by some small amount $\epsilon$, and then all results hold up to $\epsilon nq$.
Moreover, note that by rescaling valuations, we can assume they are integer-valued, which allows the clock (price) to move in integer steps.

\subsection{Results and Overview}\label{sec:results}
We begin by reviewing Ausubel's clinching auction mechanism in Section \ref{sec:ausubel}. We define the tiered auction mechanism in Section~\ref{sec:tiered}, giving two equivalent iterative and recursive characterizations of our auction. In Section \ref{sec:efficient} we study efficient outcomes in the auction, which we show are achieved by a greedy direct revelation mechanism. By comparison with this greedy algorithm, we show the efficiency of our auction in Section \ref{sec:tiered-opt}, by characterizing the valuations which clinch an item. 
To show the incentive compatibility of the tiered clinching auction, we prove in Section \ref{sec:vcg} that it charges VCG prices under sincere bidding, from which we conclude that our auction is ex-post incentive compatible in Section \ref{sec:incentives}.

\section{Ausubel's Clinching Auction Mechanism}\label{sec:ausubel}
Our tiered clinching auction mechanism is a generalization of the classical (single-tiered) clinching auction of Ausubel~\cite{Ausubel2004}.
So let us first recall Ausubel's auction. A formal pseudocode description of his ascending price mechanism is given in Auction~\ref{alg_ausubel}. Informally, at a price $p$ each bidder $i$ announces their demand $D_i(p)$. 
Bidder $i$ {\em clinches} an item for the price $p$ if the demand of all bidders except for bidder $i$ is less than the residual supply at price $p$. If the total number of items clinched is less than the supply then the price rises and the allocation process is repeated.
The price increases until the aggregate demand drops below the supply and the market clears. 

\begin{algorithm}[h!]
\caption{Ausubel's Clinching Auction}
\label{alg_ausubel}
\begin{algorithmic}[2]
\State Initialize price $p=0$. 
\State Initialize allocations $C_1 = \dots = C_n = 0$.
\Loop
    \State Ask each bidder $i$ for their demand $D_i(p)$ at the current price $p$.
    \For{$i=1$ to $n$}
    \If{{$\sum_{j \neq i}  D_j(p) < m - C_i$}}
        \State Allocate a new item to bidder $i$: 
        \State $C_i \leftarrow C_i + 1$
        \State Charge bidder $i$ price $p$ for it.
    \EndIf
    \EndFor
    \If {{$\sum_{i=1}^n D_i(p) \leq m$} }
        \State {\bf Terminate} the auction.
    \EndIf
    
    \State $p \leftarrow p + 1$
\EndLoop
\end{algorithmic}
\end{algorithm}

We remark that each bidder $i$ is required to follow a monotone activity rule. That is, they may never increase their demand after a price increase and their demand $D_i(p)$ must be at least the current number $C_i$ of items they have clinched thus far. Moreover, these properties automatically hold if all agents bid sincerely (truthfully), as one can show that the price never exceeds any clinching valuation. Note that by rescaling valuations, we can assume they are integer-valued, which allows the clock (price) to move in integer steps.

\begin{example}
\label{ex_Ausubels_auction}
\ \\ $\blacktriangleright$
    Consider an instance with a supply of two and three bidders. The Red bidder has decreasing marginal valuations of $u_R(1) = 6$, $u_R(2) = 5$, the Blue bidder has decreasing marginal valuations of $u_B(1) = 7$, $u_B(2) = 2$, and the Green bidder has unit-demand with the marginal valuations $u_G(1) = 3$, $u_G(2) = 0$. 
 
   \begin{table}[h!]
    \centering
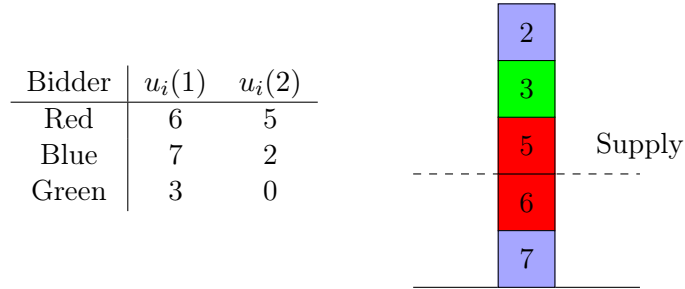
    \begin{tabular}{c|cc}
        Bidder & $u_i(1)$ & $u_i(2)$\\
        \hline
        Red & 6 & 5  \\
        Blue & 7 & 2 \\
        Green & 3 & 0 \\
    \end{tabular}
    \hspace{1cm}
    \begin{tikzpicture}[baseline=(current bounding box.center), scale=0.75]
        \def\val[#1,#2](#3,#4){
            \fill[draw=black,color=#2] (#3-0.5,#4-0.5) rectangle (#3+0.5,#4+0.5);
            \draw (#3-0.5,#4-0.5) rectangle (#3+0.5,#4+0.5);
            \node at (#3,#4) {#1}}
     
        \draw[dashed] (-2,1.5) -- (2,1.5);
            \node at (2, 2) {Supply};
        \draw (-2,-0.5) -- (2,-0.5);
        \val[6,red](0,1);
        \val[5,red](0,2);
        \val[7,blue!35](0,0);
        \val[2,blue!35](0,4);
        \val[3,green](0,3);
    \end{tikzpicture}
\end{table} 
     
    The clinching auction begins at price $p=0$. The price then starts rising. At $p=2$, Blue reduces its demand from $2$ to $1$. However, for any bidder $i$ the demand of the other two bidders is still at least the supply $m=2$ minus $C_i=0$. 
    Thus, no items are allocated at $p=2$ and the price rises to $p=3$. At this price Green reduces its demand to $0$. The total sum of demands of the Blue and Green bidders is now $1$. This is strictly less than the supply $m=2$ minus $C_R=0$. So the Red bidder wins one item for a price of $3$ (and we set $C_R=1$).  The price then continues to rise until $p=5$ when Red reduces its demand to $1$. The total demand of the Red and Green bidders is now $1$ which is less than the supply $m=2$ minus $C_B=0$. So Blue is allocated an item for the price $5$ and we set $C_B=1$. (Observe that Red is not allocated a second item at $p=5$ because $C_R$ is already $1$, so the allocation condition on Line 6 of the Auction~\ref{alg_ausubel} is violated.) \hfill $\blacktriangleleft$
\end{example}

This example illustrates several interesting properties.
First observe that the Blue bidder, despite having the highest valuation for the first item, was not allocated the first item. Indeed, the mechanism has access to demand functions, not valuation functions.
Second, the Red and Blue bidders were allocated one item each which is the efficient outcome.
Third, the Blue and Red bidders were not charged the same price for their items, so the mechanism is not envy-free. However, the reader may verify that they were charged
the VCG prices. This is not a surprise: Ausubel~\cite{Ausubel2004} showed that the clinching auction is ex-post incentive compatible and that the resultant equilibrium where everyone bids sincerely yields the efficient outcome.

\section{The Tiered Clinching Auction Mechanism}
\label{sec:tiered}
Having recalled the clinching auction for homogeneous multi-unit auctions, we will now extend it to non-homogeneous multi-unit auctions when items are tiered.
Given there are $L$ tiered types of items, intuitively these induce $L$ submarkets and our basic idea is to run $L$ simultaneous ascending clock auctions and see in which submarket an item is clinched first.

To implement this idea we have to be careful in how we define our $L$ submarkets. The key observation is that a bidder $i$ of tier $\ell(i)=t$ desires any item of quality $t$ or above.
This induces a collection of {\em nested} submarkets, denoted $\mathcal{M}_1, \mathcal{M}_2, \dots, \mathcal{M}_L$.
Here the submarket $\mathcal{M}_t$ consists of all the items of quality $t$ and below and the bidders of tier $t$ and below. 
It is important to observe that submarket $\mathcal{M}_{t-1}$ is strictly contained within submarket $\mathcal{M}_{t}$. Furthermore, depending upon its tier, a bidder $i$ in submarket $\mathcal{M}_t$ need not desire each item but rather only a subset of the items in that submarket.

\subsection{Iterative Characterization of the Tiered Clinching Auction}
\label{sec:iterative}
With the submarkets $\mathcal{M}_1, \mathcal{M}_2, \dots, \mathcal{M}_L$ defined, we may now formalize the tiered clinching auction. To simultaneously run a clinching auction in
each of these submarkets we have a single clock with a common
ascending price $p$. A formal description of the tiered clinching auction is shown in Algorithm~\ref{alg_iterative}.

\begin{algorithm}
\caption{The Iterative Tiered Clinching Auction}
\label{alg_iterative}
\begin{algorithmic}[2]
\State Initialize price $p=0$.
\State Initialize allocations $C_1 = C_2 = \dots = C_n = 0$
\Loop
    \State Ask each bidder $i$ for their demand $D_i(p)$ at the current price $p$.
    \For{$t=1, \dots, L$}
    \For{$i = 1, \dots, n$ such that $\ell(i) \leq t$}
        \If{{$\sum_{j:j \neq i, \ell(j) \leq t} \max\{D_j(p) - C_j,0\} < \sum_{k = 1}^t q_k - \sum_{j:\ell(j) \leq t} C_j$}}
        \State Allocate bidder $i$ an available item of the lowest, acceptable quality.
        \State $C_i \leftarrow C_i + 1$
        \State Charge bidder $i$ price $p$ for it.
        \EndIf
    \EndFor
    \EndFor
    \If{{$\sum_{i : \ell(i) \geq t} D_i(p) \leq \sum_{k = t}^L q_k$ for all $1 \leq t \leq L$}}
        \State \textbf{Terminate} the auction.
    \EndIf
    \State $p \leftarrow p+1$
\EndLoop
\end{algorithmic}
\end{algorithm}
Let us parse this pseudocode. At price $p$, each bidder $i$ states its demand $D_i(p)$, which is the total number of items it desires at that price. In order to determine which agent clinches an item next, we need to know how many items 
each agent has already clinched. So we let $C_i$ be the number of items that $i$ has already clinched over all submarkets. 

For a submarket $\mathcal{M}_t$, a bidder $i$ clinches at price $p$ in that submarket if 
$$\sum_{j:j \neq i, \ell(j) \leq t} \max\{D_j(p)-C_j, 0\} \ <\  \sum_{k = 1}^t q_k - \sum_{j: \ell(j) \leq t} C_j.$$
That is, the additional demand of all other bidders is less than the remaining supply in that submarket. 
This condition for a bidder to clinch in a specific submarket is analogous to Ausubel's single-tiered auction. Indeed, the clinching condition of line 6 in Auction~\ref{alg_ausubel} can be rewritten as 
$$\sum_{j \neq i} (D_j(p) - C_j) \ < \ m - \sum_{i=1}^n C_i.$$

One subtlety here is that in Ausubel's clinching auction the
price can never increase above a `clinching valuation'.
Thus the demand $D_j(p)$ is never less than $C_j$.
However, in the tiered clinching auction, the price can increase above a clinching valuation, inducing the associated bidder $j$ to demand fewer items than they have already clinched. For this reason it is necessary in line 7 to take the maximum of $D_j(p) - C_j$ and zero, which we call the \emph{residual demand} of bidder $j$, thus ensuring that their de facto demand for additional items is never negative.

We emphasize that at price $p$ different agents may simultaneously clinch in the same or different submarkets.
For this reason, in lines 5 and 6, we search for clinches by evaluating each submarket in ascending order and then by ordering the bidders in each submarket by their index numbers. The first bidder we find that clinches is allocated the next item. 
Observe that the item allocated to that bidder (line 8) is the lowest quality item that is currently available and which the bidder is willing to accept. Note this item is {\bf not} required to be in the submarket the item was clinched in!
(This observation is extremely important. Indeed a bidder of tier $t_1$ may be allocated an item from tier $t_2$ even though it clinched an item in submarket $\mathcal{M}_{t_3}$, where $t_1 < t_2 < t_3$. This fact means that when an agent clinches in the tiered clinching auction it can have unexpected consequences in different submarkets. In turn, this renders the efficiency and incentive compatibility proofs in subsequent sections quite subtle.)

The price rises (line 17) after each bidder has had a chance to clinch, unless the termination condition that all markets can clear is satisfied (line 14). Specifically, the auction ends at the point where, for every tier $t$, we have
$\sum_{i : \ell(i) \geq t} D_i(p) \leq \sum_{k = t}^L q_k$, i.e. every bidder's demands can be satisfied with the supply.

If a market is initially undersubscribed, it can immediately be cleared by assigning all items at no cost to bidders who demand them. Thus we can assume without loss of generality, that no market is undersubscribed at the beginning of the auction.

To understand the workings and intricacies of the tiered clinching auction, it will be informative to present an example here (and several more later).

\begin{example}
\label{ex_TCA}
\ \\ $\blacktriangleright$
Consider an auction with three tiers, each with a supply of two items. Let there be five bidders whose tiers and valuations are shown in Figure~\ref{ex2} (where all subsequent marginal valuations are assumed to be equal to 0).
\begin{figure}[h!]
    \centering
    \begin{tabular}{c|c|ccc}
        Bidder & Tier & $u_i(1)$ & $u_i(2)$ & $u_i(3)$\\
        \hline
        Red & 1 & 9 & 4 & 0 \\
        Blue & 1 & 5 & 2 & 0\\
        Green & 2 & 10 & 8 & 7\\
        White & 2 & 6 & 3 & 0\\
        Grey & 3 & 11 & 1 & 0\\
    \end{tabular}
    \caption{Tiers and marginal valuations of the five bidders}\label{ex2}
\end{figure}
Starting at $p=0$, the aggregate demand of bidders in tier 1 is 4,
in tier 2 is 5, and in tier 3 is 2. No bidders clinch at $0$ so the price increases to $p=1$. The demand in tier 3 now drops to 1, but still no bidder satisfies a clinching condition. The price increases to $p=2$. The situation at this point is shown in the LHS of Figure~\ref{ex2:p234}, where valuations are stacked in decreasing order from the bottom in their respective tiers and colour-coded to indicate which bidder they belong to. Valuations that are not above the current clock price (and hence no longer contribute to demand) are crossed out, and valuations that have already clinched an item include information about the quality of the item received (top left), the submarket in which it was clinched (top right), and the price at which the clinch occurred (bottom).
Observe the total demand in tier 1 has now decreased to 3. The submarket $\mathcal{M}_1$ is still oversubscribed in aggregate demand. But, the demand of all bidders other than Red is 1, which is less than the supply of 2. Therefore, Red clinches an item at the price $p = 2$ in the submarket $\mathcal{M}_1$.

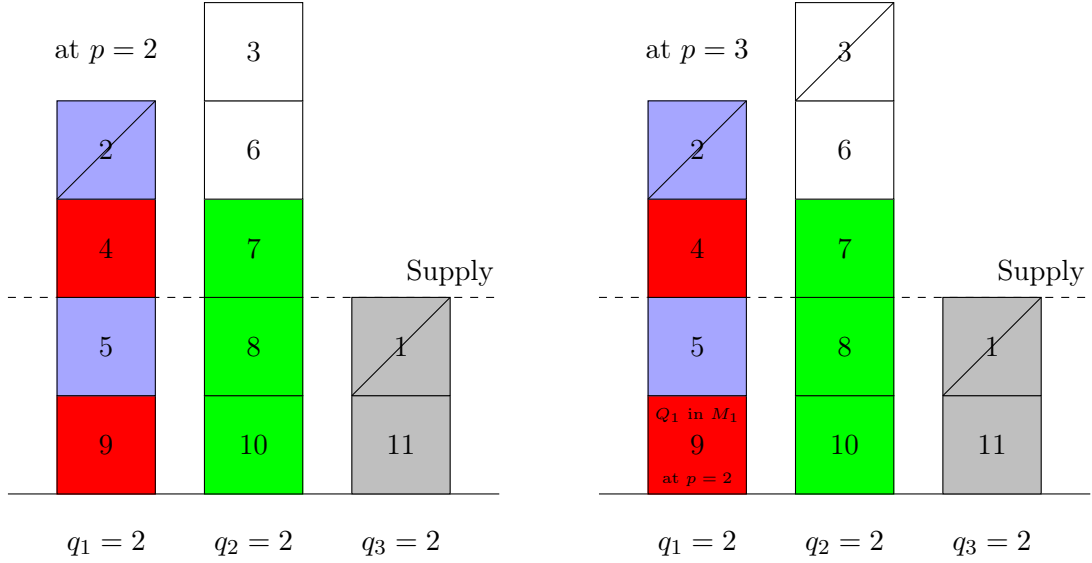
\begin{figure}[h!]
\begin{center}
    \begin{tikzpicture}[baseline=(current bounding box.center), scale=0.65]
        \draw (-1,0) -- (9,0);
        \draw[dashed] (-1,4) -- (9,4);
            \node at (8, 4.5) {Supply};
        \node at (1,-1) {$q_1 = 2$};
        \node at (4,-1) {$q_2 = 2$};
        \node at (7,-1) {$q_3 = 2$};

        \node at (1,9) {at $p=2$};
            
        \valDetailed[9,red](0,0)[in,0,0,0];
        \valDetailed[5,blue!35](0,1)[in,0,0,0];
        \valDetailed[4,red](0,2)[in,0,0,0];
        \valDetailed[2,blue!35](0,3)[out,0,0,0];
        \valDetailed[10,green](1,0)[in,0,0,0];
        \valDetailed[8,green](1,1)[in,0,0,0];
        \valDetailed[7,green](1,2)[in,0,0,0];
        \valDetailed[6,white](1,3)[in,0,0,0];
        \valDetailed[3,white](1,4)[in,0,0,0];
        \valDetailed[11,lightgray](2,0)[in,0,0,0];
        \valDetailed[1,lightgray](2,1)[out,0,0,0];
    \end{tikzpicture}
    \hspace{1cm}
    \begin{tikzpicture}[baseline=(current bounding box.center), scale=0.65]
        \draw (-1,0) -- (9,0);
        \draw[dashed] (-1,4) -- (9,4);
            \node at (8, 4.5) {Supply};
        \node at (1,-1) {$q_1 = 2$};
        \node at (4,-1) {$q_2 = 2$};
        \node at (7,-1) {$q_3 = 2$};

        \node at (1,9) {at $p=3$};
            
        \valDetailed[9,red](0,0)[clinched,1,1,2];
        \valDetailed[5,blue!35](0,1)[in,0,0,0];
        \valDetailed[4,red](0,2)[in,0,0,0];
        \valDetailed[2,blue!35](0,3)[out,0,0,0];
        \valDetailed[10,green](1,0)[in,0,0,0];
        \valDetailed[8,green](1,1)[in,0,0,0];
        \valDetailed[7,green](1,2)[in,0,0,0];
        \valDetailed[6,white](1,3)[in,0,0,0];
        \valDetailed[3,white](1,4)[out,0,0,0];
        \valDetailed[11,lightgray](2,0)[in,0,0,0];
        \valDetailed[1,lightgray](2,1)[out,0,0,0];
    \end{tikzpicture}
\end{center}
\vspace{-0.5cm}
\caption{Tiered clinching auction at price steps $p = 2$ and 3}\label{ex2:p234}
\end{figure}

No other bidders clinch, so the price increases to $p=3$. 
The situation at this point is shown in the center of Figure~\ref{ex2:p234}. At this price, no bidders clinch in the submarkets $\mathcal{M}_1$ or $\mathcal{M}_2$.
But Green now clinches an item in $\mathcal{M}_3$.
To see this, observe that the residual demand of the other four bidders in $\mathcal{M}_3$ is 4.
This is less than the remaining supply of 5 (note only one of the six items has so far been allocated). Therefore, Green clinches one item in $\mathcal{M}_3$ for the price of 3. Despite the item being clinched in $\mathcal{M}_3$, the lowest quality item that is available and acceptable to Green is from tier 2. So, Green is allocated an item from tier 2. 

The price now rises to $p=4$, as shown in the top-left of Figure~\ref{ex2:p56}, and two clinches occur. Ordered from the smallest submarket to the largest, the first clinch is in $\mathcal{M}_1$. Specifically, Red decreases its demand by $1$, causing Blue to clinch an item in $\mathcal{M}_1$. 
The second clinch occurs in $\mathcal{M}_3$. Green again clinches in $\mathcal{M}_3$ because the total residual demand of the other bidders is now 2, which is less than the remaining supply of 3. Again, as a result of this clinch, Green is allocated an item from tier 2.
The price now increases to $p=5$. Blue's demand drops below the number of items it has already clinched. As shown in the top-right of Figure~\ref{ex2:p56}, their marginal valuation of 5 previously clinched an item, but is no longer above the price clock. This does not affect the clinching conditions of the other bidders, so no items are allocated at this price.

\begin{figure}[h!]
\begin{center}
    \begin{tikzpicture}[baseline=(current bounding box.center), scale=0.65]
        \draw (-1,0) -- (9,0);
        \draw[dashed] (-1,4) -- (9,4);
            \node at (8, 4.5) {Supply};
        \node at (1,-1) {$q_1 = 2$};
        \node at (4,-1) {$q_2 = 2$};
        \node at (7,-1) {$q_3 = 2$};

        \node at (1,9) {at $p=4$};
            
        \valDetailed[9,red](0,0)[clinched,1,1,2];
        \valDetailed[5,blue!35](0,1)[in,0,0,0];
        \valDetailed[4,red](0,2)[out,0,0,0];
        \valDetailed[2,blue!35](0,3)[out,0,0,0];
        \valDetailed[10,green](1,0)[clinched,2,3,3];
        \valDetailed[8,green](1,1)[in,0,0,0];
        \valDetailed[7,green](1,2)[in,0,0,0];
        \valDetailed[6,white](1,3)[in,0,0,0];
        \valDetailed[3,white](1,4)[out,0,0,0];
        \valDetailed[11,lightgray](2,0)[in,0,0,0];
        \valDetailed[1,lightgray](2,1)[out,0,0,0];
    \end{tikzpicture}
    \hspace{1cm}
    \begin{tikzpicture}[baseline=(current bounding box.center), scale=0.65]
        \draw (-1,0) -- (9,0);
        \draw[dashed] (-1,4) -- (9,4);
            \node at (8, 4.5) {Supply};
        \node at (1,-1) {$q_1 = 2$};
        \node at (4,-1) {$q_2 = 2$};
        \node at (7,-1) {$q_3 = 2$};

        \node at (1,9) {at $p=5$};
            
        \valDetailed[9,red](0,0)[clinched,1,1,2];
        \valDetailed[5,blue!35](0,1)[clinched_out,1,1,4];
        \valDetailed[4,red](0,2)[out,0,0,0];
        \valDetailed[2,blue!35](0,3)[out,0,0,0];
        \valDetailed[10,green](1,0)[clinched,2,3,3];
        \valDetailed[8,green](1,1)[clinched,2,3,4];
        \valDetailed[7,green](1,2)[in,0,0,0];
        \valDetailed[6,white](1,3)[in,0,0,0];
        \valDetailed[3,white](1,4)[out,0,0,0];
        \valDetailed[11,lightgray](2,0)[in,0,0,0];
        \valDetailed[1,lightgray](2,1)[out,0,0,0];
    \end{tikzpicture} \\[0.8cm]
    \begin{tikzpicture}[baseline=(current bounding box.center), scale=0.65]
        \draw (-1,0) -- (9,0);
        \draw[dashed] (-1,4) -- (9,4);
            \node at (8, 4.5) {Supply};
        \node at (1,-1) {$q_1 = 2$};
        \node at (4,-1) {$q_2 = 2$};
        \node at (7,-1) {$q_3 = 2$};

        \node at (1,9) {at $p=6$};
            
        \valDetailed[9,red](0,0)[clinched,1,1,2];
        \valDetailed[5,blue!35](0,1)[clinched_out,1,1,4];
        \valDetailed[4,red](0,2)[out,0,0,0];
        \valDetailed[2,blue!35](0,3)[out,0,0,0];
        \valDetailed[10,green](1,0)[clinched,2,3,3];
        \valDetailed[8,green](1,1)[clinched,2,3,4];
        \valDetailed[7,green](1,2)[in,0,0,0];
        \valDetailed[6,white](1,3)[out,0,0,0];
        \valDetailed[3,white](1,4)[out,0,0,0];
        \valDetailed[11,lightgray](2,0)[in,0,0,0];
        \valDetailed[1,lightgray](2,1)[out,0,0,0];
    \end{tikzpicture}
    \hspace{1cm}
    \begin{tikzpicture}[baseline=(current bounding box.center), scale=0.65]
        \draw (-1,0) -- (9,0);
        \draw[dashed] (-1,4) -- (9,4);
            \node at (8, 4.5) {Supply};
        \node at (1,-1) {$q_1 = 2$};
        \node at (4,-1) {$q_2 = 2$};
        \node at (7,-1) {$q_3 = 2$};

        \node at (1,9) {at the end};
            
        \valDetailed[9,red](0,0)[clinched,1,1,2];
        \valDetailed[5,blue!35](0,1)[clinched_out,1,1,4];
        \valDetailed[4,red](0,2)[out,0,0,0];
        \valDetailed[2,blue!35](0,3)[out,0,0,0];
        \valDetailed[10,green](1,0)[clinched,2,3,3];
        \valDetailed[8,green](1,1)[clinched,2,3,4];
        \valDetailed[7,green](1,2)[clinched,3,3,6];
        \valDetailed[6,white](1,3)[out,0,0,0];
        \valDetailed[3,white](1,4)[out,0,0,0];
        \valDetailed[11,lightgray](2,0)[clinched,3,3,6];
        \valDetailed[1,lightgray](2,1)[out,0,0,0];
    \end{tikzpicture}
\end{center}
\vspace{-0.5cm}
\caption{Tiered clinching auction at price steps $p = 4, 5$, 6 and at the end}\label{ex2:p56}
\end{figure}
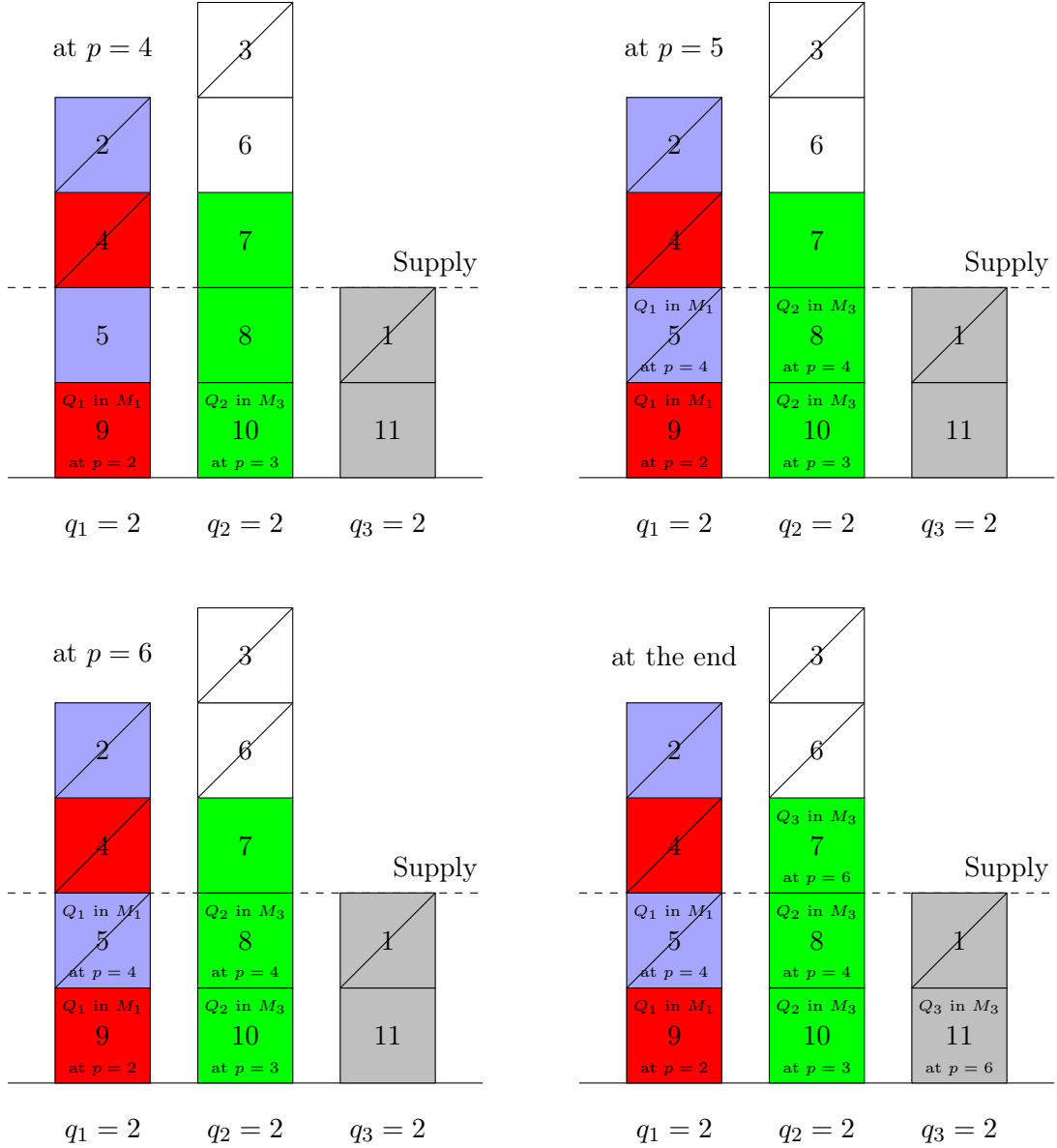
Finally, the price rises to $p=6$; see the bottom-left of Figure~\ref{ex2:p56}. There are only two items remaining, both in tier 3. Only Green and Grey 
have a remaining marginal valuation above the price of 6.
So their clinching conditions are satisfied in $\mathcal{M}_3$.
Therefore, they are each allocated an item from tier 3. 
The final outcome is illustrated in the bottom-right of Figure~\ref{ex2:p56}.
\hfill $\blacktriangleleft$
\end{example}

A few observations are in order here.
One, note that Green clinched three items in the submarket $\mathcal{M}_3$. But the resultant items it was allocated belonged to different tiers, namely two items from tier 2 and one item from tier $3$. 
Two, note that the six winning clinching valuations were not the six highest valuations; for example, White's marginal valuation of $6$ lost whilst Blue's marginal valuation of $5$ won. Regardless, one may verify that the final outcome is the efficient outcome, the feasible allocation that maximizes social welfare. We will prove this property always holds for the tiered clinching auction in Section~\ref{sec:tiered-opt}.
In order to prove this, we will be making use of a recursive formulation of the tiered clinching auction, which we show is equivalent to the iterative formulation of Auction~\ref{alg_iterative}. 

\subsection{Recursive Characterization of the Tiered Clinching Auction}
\label{sec:recursive}
A recursive formulation of the tiered clinching auction is given in
Auction~\ref{alg_recursive}. 
When a bidder clinches an item, the supply in the corresponding 
tier is decremented (line 7), and the valuation function of that
bidder is shifted (lines 9-12) to reflect the fact that its highest marginal valuation has an associated item allocated to it.
The auction then recurses until the market clears.
\begin{algorithm}[h!]
\caption{The Recursive Tiered Clinching Auction}
\label{alg_recursive}
\begin{algorithmic}[2]
\Procedure{Recursive\_TCA}{$\vec{q}$, $\vec{u}$, $p$}
    \State Ask each bidder $i$ for their demand $D_i(p)$ at the current price $p$.
    \For{$t=1, \dots, L$}
    \For{$i = 1, \dots, n$ such that $\ell(i) \leq t$}
        \If{{$\sum_{j:j \neq i, \ell(j) \leq t} D_j(p) < \sum_{k = 1}^t q_k$}}
        \State Allocate bidder $i$ an available item of the lowest acceptable quality $s$.
        \State $q_s \leftarrow q_s - 1$
        \State Charge bidder $i$ price $p$ for it.
        \For{$j = 1, \dots, \sum_{k=1}^L q_k - 1$} \Comment{Shift the utility function of bidder $i$}
            \State $u_i(j) \leftarrow u_i(j+1)$
        \EndFor
        \State $u\paren{\sum_{k=1}^L q_k} \leftarrow 0$
        \State $\Call{Recursive\_TCA}{\vec{q}, \vec{u}, p}$
        \EndIf
    \EndFor
    \EndFor
    \If{{$q_1 = \cdots = q_L = 0$}}
        \State \textbf{Terminate} the auction.
    \EndIf
    \State $p \leftarrow p + 1$ \Comment{Increase price, as no clinch occurred}
    \State $\Call{Recursive\_TCA}{\vec{q}, \vec{u}, p}$
    \EndProcedure
\end{algorithmic}
\end{algorithm}

\begin{lemma}
\label{lemma_iterative_equiv_recursive}
    The iterative tiered clinching mechanism, Auction~\ref{alg_iterative}, and the recursive tiered clinching mechanism, Auction~\ref{alg_recursive}, are equivalent.
\end{lemma}

\begin{proof}
The proof is by induction on the number of items. 

For the base case, when there is only one item, both Auctions \ref{alg_iterative} and \ref{alg_recursive} iterate through the tiers from lowest to highest and through the bidders' fixed order, while increasing the price until a clinching condition is satisfied for the first time. As no other item has yet been assigned, line 7 in Auction \ref{alg_iterative} is the same condition as line 5 in Auction \ref{alg_recursive}. After this first clinch of bidder $i$, the recursive auction terminates as there is no more supply (line 17). Similarly, the iterative auction terminates without further clinches: the demand of all other bidders is zero, and the price keeps increasing until bidder $i$'s demand drops to one (if they initially demanded more items), at which point the termination condition on line 14 is satisfied. Note that bidder $i$ cannot clinch more than one item, as the residual supply is zero.

For the induction step, suppose that the two auctions behave the same for instances with up to $m$ items. Consider an instance with $m+1$ items. The first clinch in Auctions 2 and 3 must occur at the same price, in the same submarket and by the same bidder $i$, as the clinching conditions are identical as long as no item has been assigned. Then the recursive auction shifts bidder $i$'s valuations (lines 9 to 12), which reduces their subsequent demand by one at any given price (at which they have non-zero demand). This corresponds precisely to the residual demand $D_i(p) - C_j$ computed in line 7 of the iterative auction.

At the same time, the available quantities of a tier $s$ are reduced by one (line 7) in the recursive auction. This corresponds precisely to subtracting one from the residual supply in line 7 of the iterative auction. Indeed, the only case where they differ happens when verifying whether some bidder $j$ clinched in a tier $t$ such that $l(i) \leq t < s$, as we subtract one through $C_i$, without including the supply $q_s$. In this case, the computed bound is less than the actual residual supply of that submarket, which could only prevent bidder $j$ from clinching. 

Suppose for the sake of contradiction that bidder $j$ was supposed to clinch in $\mathcal{M}_t$. Note that in this case, tiers $\ell(i)$ through $t$ have no residual supply, since bidder $i$ was assigned an item from tier $s > t$. But then, if $\ell(j) < \ell(i)$, bidder $j$ should have clinched in submarket $\mathcal{M}_{l(i)-1}$ or smaller instead, because the fact that bidder $j$ satisfies the clinching condition in tier $t$ implies that it also satisfies the condition in at least tier $\ell(i)-1$, which has the same supply. If $\ell(j) \geq \ell(i)$, then submarket $\mathcal{M}_{\ell(i)-1}$ is undersubscribed, which means all its items have been allocated, and hence there is nothing left for bidder $j$ to clinch in submarket $\mathcal{M}_t$. Having shown the equivalence of the modifications, we conclude by the induction hypothesis, that the rest of the two auctions behave the same.
\end{proof}

Our task now is to prove that the tiered clinching auction is efficient
and incentive compatible.

\section{Efficient Allocations in Tiered Auctions}\label{sec:efficient}
Our first task is to show that under sincere bidding, the tiered clinching auction outputs an efficient allocation of items. 
In order to prove this, we first have to understand what properties an efficient allocation will possess.
To do so, in this section, we describe a direct revelation mechanism which applies a greedy algorithm to optimally assign items to bidders. We will then prove in Section~\ref{sec:tiered-opt} that the tiered clinching auction outputs an optimal allocation via a comparison with the greedy allocation.

\subsection{A Direct Revelation Mechanism}
A \emph{direct revelation mechanism} for the tiered auction requires the type of each bidder $i$, information here that consists of the bidder's tier and its marginal valuation $u_i(k)$, for each quantity $1\le k\le \sum_{t=1}^L q_t$.
The efficient solution $\{x_1, \dots, x_n\}$ is then simply the allocation that maximizes
$$\sum_{i=1}^n \sum_{k=1}^{x_i} u_i(k),
\qquad \text{ subject to } 
\sum_{i : \ell(i) \geq t} x_i \leq \sum_{k = t}^L q_k \quad \forall t\in \{1,2,\dots L\},$$ 
which simply state that the set of bidders who desire items of quality at least $t$ can collectively be allocated at most the number of items
in tier $t$ and above.

Of course, if we want the direct mechanism to be incentive compatible, then we must charge each bidder a price equal to its externality. These prices are irrelevant to our objective in this section, but we will study them in Section~\ref{sec:incentives} when we prove the tiered clinching auction is incentive compatible.

\subsection{An Optimal Greedy Allocation Mechanism}\label{sec:greedy}
Assuming the direct mechanism has the truthful tier and valuation function of each bidder, it can compute an efficient allocation using the greedy allocation algorithm, presented in Auction~\ref{alg_greedy}. The idea behind the greedy algorithm is that it iterates from the lowest tier to the highest, and assigns all of the $q_t$ items of its current tier $t$ to the $q_t$ highest remaining valuations of type $t$ or lower. Then, it removes the valuations which have been allocated items by shifting the marginal utility functions of the bidders to reflect the number of items they were allocated. 

\begin{algorithm}[h!]
\caption{The Greedy Allocation Algorithm}
\label{alg_greedy}
\begin{algorithmic}[2]
\For{t = 1, \dots, L}
    \State Assign the $q_t$ items of tier $t$ to the $q_t$ highest valuations with type at most $t$.
    \For{agent $i: i$ was allocated $x_i > 0$ items in current iteration}
        \For{$j = 1, \dots, \sum_{k=1}^L q_k - x_i$} \Comment{Shift the utility function of bidder $i$}
            \State $u_i(j) \leftarrow u_i(j+x_i)$
        \EndFor
        \For{$j=0, \dots, x_i - 1$}
            \State $u\paren{\sum_{k=1}^L q_k - j} \leftarrow 0$
        \EndFor
    \EndFor
\EndFor
\end{algorithmic}
\end{algorithm}

Before proving this algorithm does indeed output an optimal allocation, let us see what happens when we run the algorithm on the instance in Example~\ref{ex_TCA}.
\begin{example}
\label{ex_greedy}
\ \\ $\blacktriangleright$
We run the greedy algorithm on the instance from Example~\ref{ex_TCA}. First, the two items from tier 1 are allocated to the highest valuations in that tier, namely 9 and 5 of agents Red and Blue respectively. Next, the two items from tier 2 are allocated to the highest valuations among those in tier 2 as well as those remaining from tier 1: Green's top valuations 10 and 8 receive the two items. Finally, the two items from tier 3 are allocated to the highest remaining valuations from all tiers, namely Grey's 11 (from tier 3) and Green's 7 (from tier 2). Notice how this greedy result is the same as the allocation computed by the tiered clinching mechanism presented in Example~\ref{ex_TCA}. We will show this equivalence in Section~\ref{sec:tiered-opt}.
$\blacktriangleleft$
\end{example}

Let us now prove that the greedy algorithm does output an efficient allocation.
Denote the allocation produced by the greedy algorithm by $G=\{x_1^g, x_2^g, \dots, x_n^g\}$, and a specific optimal solution by $O=\{x_1^*, x_2^*, \dots, x_n^*\}$.

\begin{lemma}
The number of items allocated to each bidder must be the same in $G$ and in $O$, that is, $x_i^g= x_i^*$ for each bidder $i$.
\end{lemma}
\begin{proof}
We prove this via a third allocation denoted $G' = \{ x_1^{g'}, \dots, x_n^{g'} \}$. This allocation is produced by running a modified version of the greedy algorithm that we dub  {\em capacitated greedy}.
Specifically, under the capacitated greedy algorithm no bidder $i$ can be allocated more than $x_i^*$ items. 
Thus as we run the algorithm, if a bidder $i$ reached the point where it has been allocated their optimal number of items, $x_i^*$, then its remaining marginal valuations are removed and no longer considered by capacitated greedy. 

We now proceed in two steps. First, we show the capacitated greedy algorithm outputs an efficient allocation, namely $G' = O$. Second, we prove the greedy algorithm outputs the same allocation as the capacitated greedy algorithm, namely $G' = G$. It immediately follows that the greedy algorithm is efficient.

Given the assumption of a tier-1 dummy bidder, we may assume every item is allocated under both the optimal and greedy solutions, that is $\sum_i x_i^*=\sum_i x_i^g =\sum_{k=1}^L q_k$. 
For the first step, if the modified greedy algorithm allocates every item as well, then by the fact that $x_i^{g'} \leq x_i^g$, for each bidder $i$, it must be that $G' = O$. For the sake of contradiction, suppose that $G'$ allocates strictly fewer items than $O$. Then there must be a tier $t$ such that $G'$ allocated every item in tiers $1$ to $t-1$, but did not allocate every item in tier $t$. Consequently, every bidder $i$ with $\ell(i) \leq t$ must has been allocated $x_i^*$ items. But this implies $$\sum_{i : \ell(i) \leq t} x_i^* < \sum_{k=1}^t q_k$$
Now since $\sum_i x_i^* =\sum_{k=1}^L q_k$, this implies that $\sum_{i : \ell(i) > t} x_i^* >  \sum_{k=t+1}^L q_k$.
Hence $O$ is infeasible, a contradiction. We conclude that the capacitated greedy algorithm assigns each bidder the same number of items as in the optimal allocation, that is $G' = O$.

Second, we show that $G' = G$. This follows from a simple exchange argument. Assume $G$ and $G'$ are different. We can view both the greedy and capacitated greedy algorithms as allocating the items one at a time starting from the lowest tier. Thus, there is a {\bf first time} $\tau$ at which the greedy algorithm and the capacitated greedy algorithm differ in their allocation. Let this happen in tier $1\le t\le L$. Let $j$ be the bidder who was allocated an item by capacitated greedy at time $\tau$, and $i$ be the bidder who was allocated an item by greedy at time $\tau$.

Recall that greedy and capacitated greedy use the same allocation rule until some bidder $k$'s capacity of items, $x_k^*$, has been reached. Thus, it must be the case that bidder $i$ had already been allocated $x_i^*$ items before time $\tau$. Let $ \{\hat{x}_1^g, \dots, \hat{x}_n^g \}$ be the allocation produced by greedy up to and including time $\tau$. We conclude that $\hat{x}_i^g=x_i^{g'}+1 = x_i^* + 1$.
Therefore, under the greedy algorithm, bidder $i$ is allocated at least $\hat{x}_i^g= x_i^{g'}+1$ items from tiers $1$ to $t$. We know that before time $\tau$, bidder $j$ has been allocated fewer than $x_j^*$ items by greedy and capacitated greedy. By definition of the greedy algorithm, this implies that
$u_i(\hat{x}_i^g)> u_j(\hat{x}_j^g+1)$.
But $u_j(\hat{x}_j^g+1) \ge u_j(x_j^*)$ as 
$\hat{x}_j^g \le x_j^* -1$. Consequently, 
$u_i(x_i^*+1)=u_i(\hat{x}_i^g)> u_j(x_j^*)$.
But then the solution produced by $G'$ can be improved by reallocating the item allocated at time $\tau$ from $j$ to $i$. This contradicts the optimality of $O$.
\end{proof}

Using this lemma, we are able to prove the desired result.

\begin{theorem}
\label{greedy_is_opt}
    The greedy allocation mechanism produces an allocation with optimal welfare.
\end{theorem}

\begin{proof}
We proceed by showing that the allocations $G$ and $O$ have equal welfare. The valuations of each bidder $i$ do not depend on the quality of the items they receive (except that the quality of each item must be at least $\ell(i)$). In $G$, by construction, the items allocated to each bidder are always of quality at least $\ell(i)$. If this is not the case in $O$, then since $O$ allocates the same number of items to each bidder, it must be a strictly worse solution, contradicting its optimality.
\end{proof}

\section{The Tiered Clinching Auction is Efficient}
\label{sec:tiered-opt}
In this section we prove that the tiered clinching auction outputs an allocation of maximum social welfare. Our goal is to show that the tiered clinching auction assigns the same number of items to each bidder as the greedy algorithm, which we know outputs an optimal solution by Theorem~\ref{greedy_is_opt}.
However, due to the subtle characteristics of the tiered clinching auction, some of which we encountered in Section~\ref{sec:tiered}, proving the efficiency of this auction is non-trivial.
For example, the obvious approach to prove this would be to prove
that the tiered clinching auction gives exactly the same allocation
as the greedy algorithm, in terms of both the number and qualities of items assigned to each bidder. Unfortunately, this is not the case.

\subsection{Greedy and Tiered Clinching Do Not Have Identical Outputs}
\label{sec:greedy-not-tiered}
Before presenting an example, let us consider why the tiered clinching auction does not necessarily give the same allocation as the greedy algorithm. This is because under the greedy algorithm (Auction~\ref{alg_greedy}) amongst the winning marginal valuations in a tier, the highest valuations get the lowest quality items!
This follows as the tiers are examined in increasing order with the highest valuations of at most that type winning earlier (line 2).
In contrast, in the tiered clinching auction (Auction~\ref{alg_iterative}) higher winning marginal valuations of a given type may receive higher quality items than lower winning marginal valuations of the same type. This is because they may clinch later in the auction (line 8) when there are less available items; consequently, the lowest quality acceptable item (line 8) may be of a high quality than if the bidder had clinched earlier in the
auction.
Let us see an example where this actually happens and the greedy and the tiered clinching auction give different allocations. 

\begin{example}
\label{ex_greedy_dif_auction}
\ \\ $\blacktriangleright$
Take an instance with two tiers and just one item in each tier. Let there be three bidders with tiers and valuations shown in Figure~\ref{ex:greedy_not_tiered}.
\begin{figure}[h!]
    \centering
    \begin{tabular}{c|c|cc}
        Bidder & Tier & $u_i(1)$ & $u_i(2)$ \\
        \hline
        Red & 1 & 4 & 0 \\
        Green & 1 & 3 & 2 \\
        Blue & 2 & 1 & 0 \\
    \end{tabular}
    \caption{Tiers and marginal valuations of the three bidders}\label{ex:greedy_not_tiered}
\end{figure}
In the tiered clinching auction, the green bidder clinches one item in submarket $2$ at price $p=1$. It then receives the available item of the lowest quality, which will be of quality 1. The red bidder then clinches an item in submarket $2$ at price $p=2$, which must then be of quality 2. 

However, the red bidder has a higher marginal valuation for its first item than blue, that is $u_R(1) > u_G(1)$. Thus, under the greedy algorithm, the red bidder is first assigned an item of quality 1. After that, the green bidder is assigned an item of quality 2. So, whilst the bidders receive the same number of items in greedy and tiered clinching, the quality of the allocated items have been switched between the red and green bidders.
\hfill $\blacktriangleleft$
\end{example}

We remark that in this example, it seems fairer in some ways that the red bidder won the higher quality item. Indeed, we believe it is a positive property of the tiered clinching auction that higher marginal valuations are able to win higher quality items.

\subsection{Characterizing the Optimal Solution}\label{sec:characterize}

We are almost ready to prove that the tiered clinching auction is efficient. To prove this we require one key concept, that of a {\em high (marginal) valuation}.
Here we formally define this concept and some relevant properties of valuations, which we will sometimes simply refer to as \emph{values}. Then we show that the tiered clinching auction only allocates items to bids corresponding to high valuations, and explain how this proves the mechanism is efficient. So let us begin by defining high valuations.
\begin{definition}
A marginal valuation $u_i(j)$ of bidder $i$ is {\em high in submarket} $m_t$, for some $t \geq \ell(i)$, if for every $1 \leq s \leq \ell(i)$, it is among the top 
$\sum_{k = s}^t q_k$ valuations in tiers $s$ to $t$.
\end{definition}
\begin{definition}
We say that a marginal valuation $u_i(j)$ of bidder $i$ is {\em high} if it is high in at least one submarket. 
\end{definition}
We will refer to the condition that $v$ is among the top $\sum_{k = s}^t q_k$ valuations in tiers $s$ to $t$ as the {\em highness constraint} on tiers $s$ to $t$, and abbreviate it as $v$ being among the \emph{top valuations} in tiers $s$ to $t$.

As stated, high valuations will be critical, so we need to understand their properties. 
First, observe that a marginal valuation may be high in many submarkets. That fact is unproblematic. However, much more concerning, is the fact that the set of submarkets in which a marginal valuation is high does not have any discernible pattern. For instance, let us see an example where the set of submarkets in which a marginal valuation is high is not contiguous.

\begin{example}
\label{ex_high_not_contiguous}
\ \\ $\blacktriangleright$
 Let there be three bidders and three tiers with one item each. The bidders' tiers and valuations are given in Figure~\ref{ex:not-contiguous}.
\begin{figure}[h!]
\centering
    \begin{tabular}{c|c|cc}
        Bidder & Tier & $u_i(1)$ & $u_i(2)$ \\
        \hline
        Red & 1 & 2 & 0 \\
        Green & 2 & 4 & 3 \\
        Blue & 3 & 1 & 0 \\
    \end{tabular}
    \hspace{1cm}
    \begin{tikzpicture}[baseline=(current bounding box.center)]
        \draw[dashed] (-1.5,0.5) -- (6,0.5);
            \node at (5.5,0.75) {Supply};
        \node at (0,-1) {$q_1 = 1$};
        \node at (2,-1) {$q_2 = 1$};
        \node at (4,-1) {$q_3 = 1$};
        
        \draw (-1.5,-0.5) -- (6,-0.5);
        \val[2,red](0,0);
        \val[4,green](2,0);
        \val[3,green](2,1);
        \val[1,blue!35](4,0);
    \end{tikzpicture}
    \caption{Example in which valuations are high in noncontiguous submarkets}\label{ex:not-contiguous}
\end{figure}
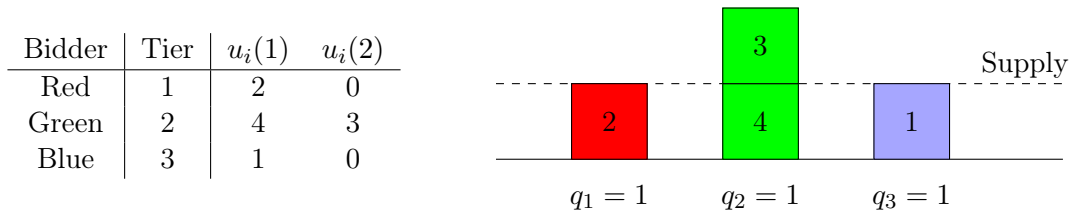
Observe the red bidder's first (and only) marginal valuation $u_1(1) = 2$ is high in its own submarket $\mathcal{M}_1$. It is not, however, high in submarket $\mathcal{M}_2$. This is because the green bidder has two higher marginal valuations and $\sum_{k = 1}^2 q_k = 1+1=2$, so it fails the corresponding highness constraint. On the other hand, the marginal valuation $u_1(1) = 2$ is high in submarket $\mathcal{M}_3$. This is because it is the third highest marginal valuation $\mathcal{M}_3$ and $\sum_{k = 1}^3 q_k = 1+1+1=3$, so it satisfies the corresponding highness constraint. Consequently, the set of submarkets in which $u_1(1)$ is high, namely $\{1,3\}$, is not a contiguous set.
\hfill $\blacktriangleleft$
\end{example}

So what structure can we infer about the high values?
Define $H_t$ to be the high values in $\mathcal{M}_t$.
Define $H^*_t$ to be values that are high for the first time in $\mathcal{M}_t$, that is, values that are high in submarket $\mathcal{M}_t$ but not high in the submarkets $\mathcal{M}_1,\mathcal{M}_2,\dots, \mathcal{M}_{t-1}$. The set $H^*_t$ has the following nice property.
\begin{lemma}
\label{lemma_char_H_star}
$H^*_t$ consists of the $q_t$ highest valuations in $\mathcal{M}_t$ that were not previously high in any of $\mathcal{M}_1,\dots, \mathcal{M}_{t-1}$.
\end{lemma}

\begin{proof}
By induction on the number of tiers. For the base case, consider the first tier $t=1$. The high valuations, $H_1$, are exactly the highest $q_t$ valuations. 
Trivially, these valuations also form $H_1^*$.
Thus, the base case holds.

Next, consider a tier $t > 1$ and assume that the induction hypothesis holds for all lower tiers. Let the values from tiers $1$ to $t$ that were not previously high in $\mathcal{M}_1, \dots, \mathcal{M}_{t-1}$ be denoted $v_1, \dots, v_l$, in descending order of value.
Our task is to prove $H_t^*=\{v_1, v_2, \dots, v_{q_t}\}$.
We break the proof into three parts. First, we will prove that the newly high valuations form a continuous sequence in this ordering, that is, $H_t^*=\{v_1, v_2, \dots, v_{r}\}$, for some $r\le l$. Second, we will prove that $r\ge q_t$. Third, we will prove that $r\le q_t$.

First let us verify $H_t^*=\{v_1, v_2, \dots, v_{r}\}$, for some $r\le l$. Suppose not. Then there exists $v_i \notin H^*_t$ but $v_{i+1} \in H^*_t$, for some $1 \leq i < l$. Abusing the notation for a bidder's tier, let $\ell(v)$ denote the tier of the bidder to whom the valuation $v$ belongs. If $\ell(v_i) \leq \ell(v_{i+1})$ then any highness constraint that applies to $v_i$ must also apply to $v_{i+1}$. But then, as $v_i > v_{i+1}$, it must be that $v_i$ is also high, a contradiction.

So assume $\ell(v_i) > \ell(v_{i+1})$. Because $v_i$ is not high in $\mathcal{M}_t$, it must violate some highness constraint, say on tiers $s$ to $t$, where $\ell(v_{i+1}) < \ell(v_i) \le s \leq t$. Thus there are $\sum_{k=s}^t q_k$ valuations in tiers $s$ to $t$ that beat $v_i$. Since $v_i> v_{i+1}$, these valuations also beat $v_{i+1}$. But $v_{i+1}\in H^*_t$ is high in submarket $m_t$. So, for any $1 \leq w \leq \ell(v_{i+1})$, it is in the top $\sum_{k=w}^t q_k$ valuations for tiers $w$ to $t$. Therefore, as there are $\sum_{k=s}^t q_k$ values in tiers $s$ to $t$ that beat $v_{i+1}$, it must be that $v_{i+1}$ is in the top $\sum_{k=w}^{s-1} q_k$ valuations for tiers $w$ to $(s-1)$. Consequently, $v_{i+1}$ was actually high in the earlier submarket $\mathcal{M}_{s-1}$, contradicting the fact that
$v_{i+1}\in H^*_t$.

Thus $H_t^*=\{v_1, v_2, \dots, v_{r}\}$, for some $r\le l$.
Our second task is to show that $r\ge q_t$, that is there are at least $q_t$ newly high valuations. Suppose for the sake of contradiction that $v_{q_t}\notin H^*_t$.
In particular, $v_{q_t}$ is not high in submarket $\mathcal{M}_t$. Hence, there is some highness condition that $v_{q_t}$ fails to satisfy. So assume $v_{q_t}$ is not among the top $\sum_{k=s}^t q_k$ valuations in the tiers $s$ to $t$, for some $s \leq \ell(v_{q_t})$. 
Consider the valuations that could be higher than $v_{q_t}$ in the tiers $s$ to $t$. These consist, first, of all elements of $\bigcup_{k=s}^{t-1} H^*_k$ that are in tiers $s$ to $t$. There are at most $\sum_{k=s}^{t-1} q_k$ such values, by the induction hypothesis. Additionally, among all remaining valuations, by definition, the valuations $v_1, \dots, v_{q_t-1}$ are greater than $v_{q_t}$. Thus, in total there are a maximum of $\sum_{k=s}^t q_k - 1$ valuations higher than $v_{q_t}$ in tiers $s$ to $t$. This implies that $v_{q_t}$ is among the top $\sum_{k=s}^t q_k$ valuations in the tiers $s$ to $t$, a contradiction.

Thus $H_t^*=\{v_1, v_2, \dots, v_{r}\}$, for some $q_t\le r\le l$.
So third, and finally, it remains to show that $r\le q_t$.
It suffices to prove $v_{q_t + 1}$ is not high in submarket $\mathcal{M}_t$. For a contradiction, assume $v_{q_t + 1}\in H_t^*$. The following property a tier may have will be useful to us.

\begin{definition}
A tier $s$ is {\em $t$-tight} in high valuations for some tier $t \geq s$ if all of the valuations which were newly high in each of the tiers $s$ to $t$ are located in tiers $s$ to $t$, and not in any lower tiers. That is,
\[ \sum_{k=s}^{t} |H_k^* \cap \{ v : s \leq \ell(v) \leq t \}| = \sum_{k=s}^{t} |H_k^*|.\]
\end{definition}

Let $s = \ell(v_{q_t + 1})$, for brevity of notation.
Let $w$ be the largest tier that is $t$-tight in high valuations, while also satisfying $w \leq s$. Observe such a tier exists, because $w = 1$ satisfies the necessary properties. We claim that $v_{q_t + 1}$ violates the highness condition for tiers $w$ to $t$.

Consider the case where $w = s$. Now, by the induction hypothesis, any valuation that is in $H_k^*$, for $w \leq k < t$ is larger than $v_{q_t + 1}$. Furthermore, we have seen that 
$\{v_1, v_2, \dots, v_{q_t}\}\subseteq H_t^*$.
Therefore, by the tightness of tier $w$ and the induction hypothesis, there are at least $\sum_{k=w}^{t-1} |H_k^*|+q_t = \sum_{k=w}^t q_k$ valuations in tiers $w$ to $t$ greater than $v_{q_t + 1}$. As a result, $v_{q_t + 1}$ does not satisfy the highness condition for tiers $w$ to $t$, contradicting the fact it is high in $\mathcal{M}_t$.

Next consider the case where $w < s$. We construct a sequence of valuations which are monotonically increasing in value and belong to decreasing tiers.
Denote $v_{q_t + 1}$ as $x_0$, the first point of the sequence, and let $s_0 = s$ be the corresponding tier. We associate to $x_0$ the set 
$$G_0 = \bigcup_{k=s_0}^{t-1} H_k^* \cup \{v_1, \dots, v_{q_t}\}$$
of valuations which, by the induction hypothesis, must be higher than $x_0$. Not all of these are located in tiers $s_0$ to $t$, otherwise $s_0$ would be $t$-tight in high valuations, contradicting the maximality of $w$.  
However, all of these valuations must be in tier $w$ or above, otherwise $w$ would not be $t$-tight in high valuations. 

We construct the next point in the sequence, taking $x_1$ to be an element from the lowest tier $s_1$ of valuations in $G_0$. We associate to $x_1$ a set 
$G_1 = \bigcup_{k=s_1}^{s_0-1} H_k^*$
of valuations which we know to be larger than $x_1$. And, since $x_1 > x_0$, all elements of $G_1$ are also greater than $x_0$.

We proceed inductively with this construction, building a sequence $(x_0, x_1, \dots, x_i)$ satisfying $x_0 < x_1 < \dots < x_i$ and $s_0 > s_1 > \dots > s_i$, to which we associate 
\[ G_j = \bigcup_{k=s_j}^{s_{j-1} -1} H^*_k \]
for each $1 \leq j \leq i$. As long as there are elements from $G_i$ which are below tier $s_i$, we choose the next valuation $x_{i+1}$ to be from the lowest tier among valuations in $G_i$. This clearly satisfies $s_{i+1} < s_i$, and $x_{i+1} > x_i$ holds due to the following property: 
\begin{claim}
All elements in $G_j$ are strictly greater than $x_j$, for all $0 \leq j \leq i$. 
\end{claim}
\begin{proof}
\renewcommand{\qedsymbol}{$\square$}
Because $x_j$ was selected from $G_{j-1} = \cup_{k=s_{j-1}}^{s_{j-2}-1} H^*_k$, it became newly high in some tier between $s_{j-1}$ and $s_{j-2}-1$. In particular, it was not high in any tier below $s_{j-1}$. As $x_j$ belongs to tier $s_j < s_{j-1}$, by the induction hypothesis, it must have smaller value than the valuations in $H^*_k$ for all $s_j \leq k < s_{j-1}$, which is the definition of $G_j$.
\end{proof}

This inductive construction terminates when all elements of $G_i$ are in tiers $s_i$ and above, implying that tier $i$ is $t$-tight. Thus $s_i = w$ by our maximal choice of $w$. At this point, the resulting sequence allows us to show that $x_0$ is not high in $t$. 
Indeed, by construction, we have that
\[\bigcup_{k=0}^i G_k = \bigcup_{k=w}^{t-1} H^*_k \cup \{v_1, \dots, v_{q_t}\},\]
which are all larger than $x_0$, by the above claim and the fact that $x_0 < x_1 < \dots < x_i$. As this is a union of disjoint sets, applying the induction hypothesis, we obtain that
$$\left\vert \bigcup_{k=w}^{t-1} H^*_k \cup \{v_1, \dots, v_{q_t}\} \right\vert = \sum_{k=w}^{t-1} |H_k^*| + q_t = \sum_{k=w}^t q_k $$
Finally, by the fact that $w$ is $t$-tight, all of these valuations are located in tiers $w$ to $t$. Thus, there are at least $\sum_{k=w}^t q_k$ valuations in tiers $w$ to $t$ greater than $x_0 = v_{q_t + 1}$. Thus $v_{q_t + 1}$ fails highness condition for tiers $w$ to $t$. Thus $v_{q_t + 1}$ is not high in $\mathcal{M}_t$.
Thus $r\le q_t$, as desired.

This completes the proof that $H^*_t= \{ v_1, \dots, v_{q_t} \}$.
\end{proof}

As a first result, this gives that the number of high valuations corresponds to the supply.
\begin{corollary}
\label{cor_high_equals_supply}
    For any $1 \leq t \leq r$, the number of items which are high in any submarket smaller than or equal to $\mathcal{M}_t$ is exactly equal to the total supply in submarket $\mathcal{M}_t$, namely $\sum_{k=1}^t q_k$.
\end{corollary}

\begin{proof}
The items which are high in submarkets up to $\mathcal{M}_t$ are precisely $\bigcup_{k=1}^t H_k^*$, a union of disjoint sets.  By Lemma \ref{lemma_char_H_star}, $|H_k^*| = q_k$, for all $1 \leq k \leq L$, and so
\[| \bigcup_{k=1}^t H_k^* | = \sum_{k=1}^t |H_k^*| = \sum_{k=1}^t q_k. \qedhere \]
\end{proof}

Observe the similarity in the structure of the greedy algorithm and the characterization of $H_t^*$, which gives the following corollary.
\begin{corollary}
\label{cor_greedy_outputs_high}
    The greedy algorithm allocates items exactly to the set of high valuations.
\end{corollary}
\begin{proof}
    The proof is by induction on $L$, the number of tiers. When there is only a single tier, the high valuations are simply the $q_t$ highest, which corresponds to the ones chosen by the greedy algorithm. For the induction step, suppose that the greedy algorithm outputs exactly the high valuations up to tier $t-1$. Then, in the next step of the algorithm, it selects the $q_t$ highest valuations not previously selected, which, by Lemma \ref{lemma_char_H_star}, corresponds exactly to the newly high valuations $H^*_t$.
\end{proof}

\subsection{The Tiered Clinching Auction Only Assigns Items to High Valuations}

So we have seen that the greedy algorithm exactly allocates items to the high valuations. Our task now is to show that this is also the case for the tiered clinching auction.

\begin{lemma}
\label{lemma_clinch_is_high}
    Each valuation which clinches an item in the tiered clinching auction is high in the submarket in which it clinches.
\end{lemma}
\begin{proof}
The proof proceeds by induction, using the recursive characterization of the auction. It has two steps. First, we show that the first clinch in an auction is always by a high valuation. Second, we show that the recursive step which creates a new auction instance preserves the set of high valuations. This is subtle, because, as illustrated by the following Example \ref{ex_dif_high_clinch}, it is possible that removing a valuation $v_1$ causes another valuation $v_2$ to no longer be high in a particular submarket. However, as we show, it must be that $v_2$ remains high in \emph{some} submarket after the removal of $v_1$ if it was high in a submarket before.

Let $\mathcal{M}_t$ be the submarket in which the first clinch occurs. For bidder $i$ to clinch an item at the current price $p$, we must have, by definition of the auction, that 
\[ \sum_{\substack{j \neq i\\ \ell(j) \leq t} } D_j(p) < \sum_{k = 1}^t q_k. \]
This means that bidder $i$'s largest valuation $u_i(1)$ is in the top $\sum_{k = 1}^t q_k$ valuations of tiers $1$ to $t$. Suppose, for the sake of contradiction, that the valuation $u_i(1)$ is not high in $\mathcal{M}_t$, i.e. it does not belong to the top $\sum_{k = s}^t q_k$ valuations of tiers $s$ to $t$, for some $1 < s \leq \ell(i)$. Then, at the given price $p < u_i(1)$, the demand of other agents of tiers $s$ to $t$ exceeds the supply in those tiers. Formally,
\[ \sum_{\substack{j \neq i\\ s \leq \ell(j) \leq t} } D_j(p) \geq \sum_{k = s}^t q_k \]
Combining these two inequalities, we have that
\begin{align*}
    \sum_{\substack{\ell(j) < s} } D_j(p) 
    &= \sum_{\substack{j \neq i\\ \ell(j) \leq t} } D_j(p) - \sum_{\substack{j \neq i\\ s \leq \ell(j) \leq t} } D_j(p)
    < \sum_{k=1}^t q_k - \sum_{k = s}^t q_k 
    = \sum_{k = 1}^{s-1} q_k
\end{align*}
where the first summation does not need to specify $j \neq i$, since $\ell(i) \geq s$. But this implies some other agent should have clinched in the smaller submarket $\mathcal{M}_{s-1}$ before the clinch occurred in $\mathcal{M}_t$, a contradiction.
Thus $u_i(1)$ is high in $\mathcal{M}_t$ and we have completed the  first step of the proof.

It remains to show that the set of high valuations does not change after applying a recursive step of the tiered clinching auction. 
Note, first, that we may assume that any market has $q_t > 0$, otherwise we can combine those empty tiers with the next higher non-empty tier. To justify this, we show that the set of valuations which are high in any submarket is the same whether empty tiers are merged or not. Take a set of adjacent tiers that are empty and denote them $w_1, \dots, w_n$, for some $n \geq 1$. Let $t_1$ be the largest non-empty tier less than $w_1$ and $t_2$ the smallest nonempty tier greater than $w_n$.
We remark that $t_1$ and $t_2$ don't necessarily exist. If $t_1$ does not exist, we will take $t_1 = 1$. If $t_2$ does not exist, then no valuations in tiers $w_1$ and higher can be allocated items, so these tiers can be removed entirely. Let $t_2'$ denote the new tier which is created after moving all valuations in $w_1$ to $w_n$ into $t_2$.

By case analysis, a high valuation $v$ remains high regardless of whether we combine empty tiers. To show this, let $s$ be the tier in which $v$ is located. In the first case, suppose $s \leq t_1$. If $v$ is not high in a tier $\geq w_1$ then the corresponding highness conditions are unaffected by combining the empty tiers from $w_1$ up. Suppose that it is high in $w_k$ for $1 \leq k \leq n$. This implies that $v$ must be high in $t_1$ as well, since the highness conditions for $w_k$ are strictly more difficult to satisfy than the highness conditions for $t_1$. Thus $v$ is high in $t_1$ regardless of whether the empty tiers are combined. 

If $v$ is high in some tier $t_3 \geq t_2$, then it will be high in $t_3$ regardless of whether empty tiers are combined or not, since the two sets of conditions involve precisely the same sets of valuations and quantities of supply.

For the second case, suppose $t_1 < s$. Then it is only possible for $v$ to be high in a tier $t_3 \geq t_2$. Regardless of the value of $s$, by the definition of high, we know that $v$ is among the top valuations in tiers $w_1$ to $t_3$. Therefore, it is among the top valuations in tiers $t_2'$ to $t_3$, because this involves the same set of valuations and the same quantity of supply, and any highness conditions which include the tier $t_2'$ still hold as well.

Thus, for the first clinch, we can assume that when a bidder with valuation $v_1$ clinches in a submarket $m_{t_1}$, it receives an item from its own tier $s_1$. We want to show that the valuations which are high before and after a clinch are the same. A subtlety in applying recursion (iteration) arises from the fact that valuations might not stay high in all tiers that they were high in before the clinching occurred; Example~\ref{ex_dif_high_clinch} below illustrates this fact. However, a value cannot become high in new tiers. This suggests the converse approach of looking at the tiers in which valuations are high after the clinch, and showing that they must have been high in the same tier before the clinch.

So, suppose a valuation $v_2$ in tier $s_2$ is high in $t_2$ after the clinch. That is, it is in the top valuations for tiers $w$ to $t_2$, for all $1 \leq w \leq s_2$. For any $w > s_1$, or in the case $t_2 < s_1$, removing the clinching valuation $v_1$ and the assigned item from $s_1$ does not affect the highness conditions on $v_2$. In the remaining case, $w \leq s_1 \leq t_2$, the number of items available in tiers $w$ to $t_2$ was greater by one before the clinch, thus adding back the valuation $v_1$ is not sufficient to dethrone $v_2$ from belonging to the top valuations before the clinch. This proves that removing a valuation and its assigned item after a clinch does not allow new valuations to become high. Finally, since by Corollary~\ref{cor_high_equals_supply} the number of high valuations is equal to the supply, we conclude that all previously high valuations remain high after a clinch occurs.
This completes the second step of the proof.
\end{proof}

\begin{example}
\label{ex_dif_high_clinch}
\ \\ $\blacktriangleright$    
There are two tiers with one item in each. The valuations of the three bidders are shown in Figure~\ref{ex5}. 
\begin{figure}[h!]
\centering
    \begin{tabular}{c|c|cc}
        Bidder & Tier & $u_i(1)$ & $u_i(2)$ \\
        \hline
        Red & 1 & 4 & 0 \\
        Green & 1 & 3 & 2 \\
        Blue & 2 & 1 & 0 \\
    \end{tabular}
    \hspace{1cm}
    \begin{tikzpicture}[baseline=(current bounding box.center)]
        \draw[dashed] (-1.5,0.5) -- (4,0.5);
            \node at (3.5,0.75) {Supply};
        \node at (0,-1) {$q_1 = 1$};
        \node at (2,-1) {$q_2 = 1$};
            
        \draw (-2,-0.5) -- (4,-0.5);
        \val[4,red](0,0);
        \val[3,green](0,1);
        \val[2,green](0,2);
        \val[1,blue!35](2,0);
    \end{tikzpicture}
\caption{Example in which valuations are no longer high in some submarket after a clinch}\label{ex5}
\end{figure}
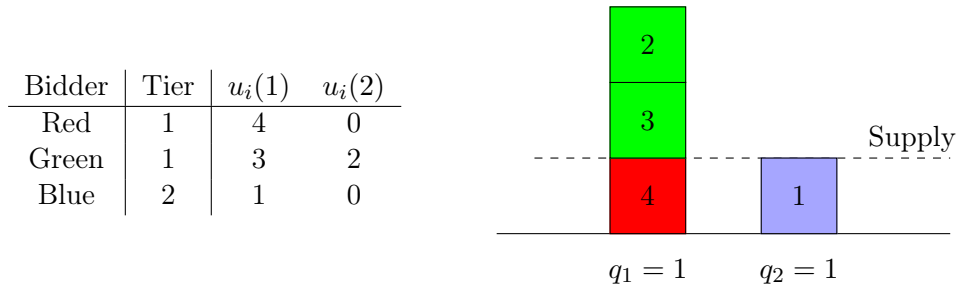
In this example, the first clinch occurs at price $p=1$, when the blue bidder's demand drops to zero. The green bidder then clinches her first item. She gets assigned an item from tier 1, the lowest available item of acceptable quality. Before the clinch, the red bidder's first valuation $u_R(1) = 4$ is clearly high in submarket $\mathcal{M}_1$. But after the clinch there are no more items left in that tier; it only remains high in the larger submarket $\mathcal{M}_2$.
\hfill $\blacktriangleleft$
\end{example}

Combining this Lemma \ref{lemma_clinch_is_high} with our results about the greedy algorithm and the number of high valuations, we obtain our desired main result.
\begin{theorem}
\label{theorem_auction_efficiency}
    If the bidders bid sincerely then the allocation computed by the tiered clinching auction is efficient.
\end{theorem}

\begin{proof}
    By Corollary \ref{cor_greedy_outputs_high}, the greedy algorithm allocates items exactly to the set of high valuations. By Lemma \ref{lemma_clinch_is_high}, any clinching valuation must be high. By Corollary \ref{cor_high_equals_supply}, the number of high valuations equals the supply and, thus, the number of clinching valuations. This proves that the high valuations are precisely the ones assigned an item by the tiered clinching auction. We conclude that the auction and the greedy algorithm give the same assignment. Since, by Theorem \ref{greedy_is_opt}, the greedy algorithm outputs an optimal solution, the tiered clinching auction is also efficient if the bidders bid sincerely.
\end{proof}

\section{VCG Payments}\label{sec:vcg}
In this section we show that, under sincere bidding, the total price each bidder pays in the tiered clinching auction is exactly their VCG payment.  We first characterize the VCG payments that would be charged by the direct mechanism, and then show how the sum of the prices a bidder pays for each of the items it wins equates to its VCG
payment.
 
The VCG payment for bidder $i$ under a direct mechanism is equal to the harm that bidder causes the other players by participating. 
Recall, by Theorem~\ref{greedy_is_opt}, that the greedy algorithm outputs an optimal solution. Therefore, running the greedy algorithm, with and without bidder $i$, to produce allocations $G$ and $G_{-i}$ can be used to compute $i$'s VCG payment. However, we also saw in Corollary \ref{cor_greedy_outputs_high} that greedy allocates items exactly to the set of high valuations. Consequently, we can characterize the difference between the welfare of $G_{-i}$ and $G$ (discounting bidder $i$'s welfare) as the sum of valuations which are high only in bidder $i$'s absence.

Let $H$ be the set of high valuations excluding any high valuation belonging to bidder $i$. Now let $H_{-i}$ be the set of high valuations in the auction where bidder $i$ does not participate. Then $V_i = H_{-i} \sm H$ is the set of valuations which are high only in bidder $i$'s absence. Observe that $V_i = \es$ if and only if bidder $i$ had no high valuations. We see that the sum of valuations in $V_i$ is equal to bidder $i$'s VCG payment. 
We now show that the prices paid by bidder $i$ for their items correspond precisely to the valuations in $V_i$.

\begin{theorem}
\label{tca_vcg_price}
If all bidders bid sincerely, then each bidder is charged their externality in the tiered clinching auction.
\end{theorem}

\begin{proof}
We must show that the sum of prices paid by bidder $i$ for their items is equal to the VCG payment they are charged by the direct VCG mechanism. Let the set $V_i$ be as above. We show that each time $i$ clinches an item, the price at which the item is clinched corresponds uniquely to a distinct valuation in $V_i$.

Using the recursive characterization of the auction, as in Lemma~\ref{lemma_clinch_is_high}, it is sufficient to consider the first clinch in the auction. Suppose this is by bidder $i$ at price $p$ in submarket $\mathcal{M}_t$. This implies 
$$\sum_{j: j \neq i, \ell(j) \leq t} D_j(p) < \sum_{k = 1}^t q_k$$ 
If this first clinch occurs at $p = 0$ then no other bidder had demand for the item, so bidder $i$'s externality in receiving it was 0, and the price matches their VCG price for the item.

Otherwise, the clinch occurs at a price $p > 0$. Since a clinch did not occur at $p-1$, it must be that some bidder $b$ decreased their demand between the steps $p-1$ and $p$. We have assumed that bidders' valuations are unique and that they are bidding sincerely. So bidder $b$ has some valuation $v$ which is equal to $p$.

Observe that $v$ is not high in the auction which includes bidder $i$. This follows from the fact that $v$ has dropped out of the auction prior to the first clinch and, by Lemma~\ref{lemma_clinch_is_high} and Corollary~\ref{cor_high_equals_supply}, the high valuations are precisely those which are assigned items by the tiered clinching auction.

However, we will now show that $v$ is high in $\mathcal{M}_t$ when we exclude from the auction all valuations belonging to bidder $i$. The condition for a clinch to occur implies $v$ must have been among the top $\sum_{k = 1}^t q_k$ valuations among bidders of tier 1 to $t$, not including bidder $i$. If $\ell(v) = 1$, this is sufficient.

Otherwise, $\ell(v) = s$ where $2 \leq s \leq t$. Suppose that $v$ does not satisfy the highness condition for tiers $w$ to $t$, where $2 \leq w \leq s$. By sincere bidding this implies that at the price $p$, \[ \sum_{j : j \neq i, w \leq \ell(j) \leq t} D_j(p) \geq \sum_{k=w}^t q_k \] But then, in order for the clinching condition to be satisfied, we have \[ \sum_{j: j \neq i, 1 \leq \ell(j) < w} D_j(p) < \sum_{k=1}^{w-1} q_k \]

\noindent
However, since $w \leq s$, in the submarket $\mathcal{M}_{w-1}$, this inequality already held at $p-1$. This implies that a clinch should already have occurred in that submarket, a contradiction. We conclude that $v$ is high in $\mathcal{M}_t$, excluding bidder $i$. Since $v$ is not high when bidder $i$'s valuations are included, this means that $v \in V_i=H_{-i}\sm H$.

We have shown that the price bidder $i$ is charged for the item they clinch is equal to a valuation in $V_i$. Moreover, since $b$ has now reduced their demand, the marginal valuation they had for this item no longer affects the auction, and future prices will not be determined by this valuation. From this we conclude that bidder $i$'s prices correspond uniquely to valuations of bidders that would have received an item in their absence. The total price that $i$ pays will be equal to the sum of the valuations in $V_i$, which is equal to their VCG payment, or externality.
\end{proof}

\section{Incentives in the Tiered Clinching Auction}\label{sec:incentives}

A mechanism is \emph{ex-post incentive compatible} if sincere bidding is a weakly dominant strategy assuming all other agents bid truthfully.

\begin{theorem}
The tiered clinching auction is ex-post incentive compatible.
\end{theorem}

\begin{proof}
Fix a bidder $i$ and assume that all bidders excluding $i$ bid sincerely. This implies that the bids of other bidders are independent of bidder $i$'s demand. Since the prices paid by bidders in the tiered clinching auction depend only on the demand of the other bidders, the prices that bidder $i$ pays are independent of $i$'s actions. Hence, agent $i$'s bids can only influence which items they are allocated, but not at what price. We have shown that the tiered clinching auction gives an efficient allocation and charges VCG prices when all bidders bid sincerely, from which we conclude that bidding sincerely maximizes $i$'s payoff.

Indeed, by bidding sincerely, bidder $i$ receives precisely all items for which they have positive utility. If they were assigned an additional item, the social welfare would decrease. This means that the bidder receiving one less item had a higher valuation for the item than bidder $i$. This other bidder's valuation is the price charged to bidder $i$, which implies that bidder $i$ has negative utility for the item at the given price. Similarly, if bidder $i$ was assigned one less item, the decrease in social welfare means that they had positive utility for it. Thus, bidder $i$ has no incentive to deviate from truthful bidding.
\end{proof}

This result extends to the setting where bidders are not given any bid information: as they are not able to respond to each other's strategies, sincere bidding is a weakly dominant strategy. Ausubel \cite{Ausubel2004} defines a more transparent setting with full bid information, private values, and incomplete information: bidders are aware of the distribution from which their opponents' utilities are drawn, but not their realization. Given this setting and the assumption that the utility function distributions have full support up to some maximum valuation, he shows that sincere bidding is a weakly dominant strategy using iterated elimination. This result equally holds for the tiered clinching auction.

%
%
\bibliographystyle{splncs04}
\bibliography{references}

\end{document}